\documentclass[9pt, preprint]{sigplanconf}

\usepackage{amsmath}
\usepackage{amsthm}
\usepackage{amssymb}
\usepackage[noend]{algorithmic}
\usepackage{algorithm}
\usepackage{amsmath}
\usepackage{fixltx2e}
\usepackage{color}
\usepackage{epsfig}
\usepackage{float}
\usepackage{citesort}
\usepackage{mathptmx}
\usepackage{url}
\usepackage{multicol}
\usepackage{multirow}
\usepackage{xspace}

\newtheorem{theorem}{Theorem}
\renewenvironment{proof}{\noindent {\bf Proof.}\ }{\qed\par\vskip 4mm\par}

\floatstyle{ruled}
\floatname{algorithm}{Procedure}

\begin{document}

\title{Memory Hierarchy Sensitive Graph Layout}
\authorinfo{Amitabha Roy\\University of Cambridge}{amitabha.roy@cl.cam.ac.uk}

\maketitle

\begin{abstract}
Mining large graphs for information is becoming an increasingly important
workload due to the plethora of graph structured data becoming available. An
aspect of graph algorithms that has hitherto not received much interest is
the effect of memory hierarchy on accesses. A typical system today has multiple
levels in the memory hierarchy with differing units of locality; ranging
across cache lines, TLB entries and DRAM pages. We postulate that it is possible
to allocate graph structured data in main memory in a way as to improve the
spatial locality of the data. Previous approaches to improving cache locality
have focused only on a single unit of locality, either the cache line or virtual
memory page. On the other hand cache oblivious algorithms can optimise layout
for all levels of the memory hierarchy but unfortunately need to be specially
designed for individual data structures. In this paper we explore hierarchical
blocking as a technique for closing this gap. We require as input a
specification of the units of locality in the memory hierarchy and lay out the
input graph accordingly by copying its nodes using a hierarchy of breadth first
searches. We start with a basic algorithm that is limited to trees and then
extend it to arbitrary graphs. Our most efficient version requires only a
constant amount of additional space. We have implemented versions of the
algorithm in various environments: for C programs interfaced with macros, as an
extension to the Boost object oriented graph library and finally as a
modification to the traversal phase of the semispace garbage collector in the Jikes
Java virtual machine. Our results show significant improvements in the access
time to graphs of various structure.
\end{abstract}

\section{Introduction}
\label{sec:intro}
Modern computer systems usually consist of a complex path to memory. This is
necessitated by the difference between the speed of computation and that of
accessing memory, often referred to as the memory wall. While microprocessor
performance has increased by 60\% every year, memory systems have increased in
performance by only 10\% every year. 

The typical solution employed by memory designers is to use faster smaller
caches to cache data from larger but slower levels of memory. For example, a
typical CPU cache would cache 64 byte lines from main memory while a Translation
Lookaside Buffer (TLB) would cache mappings for 4KB chunks of virtual address
space. A typical access to memory therefore has to negotiate many levels of
hierarchy. Locality therefore has an important role to play for the in-memory
processing of large datasets. If accesses are clustered (blocked) on the same 64
byte or 4KB chunk of memory (which we call units of spatial locality), it will
lead to fewer transfers between levels in the memory hierarchy and consequently
better performance.

Graphs form an important and frequently used abstraction for the processing of
large data. This is more so today, with increasing interest in mining graph
structured data: common examples being page-ranking that examines the
hyperlinking between web-pages, community detection in social networks,
navigational queries on road-network data or simulating the spread of epidemics
(viruses) over human (computer) networks. Thus far, little attention has been
paid to mitigating the impact of the memory hierarchy on processing large
graphs. This paper makes the case that sensitivity to the memory hierarchy can
make a big difference to the costs of processing large graphs.  

Existing research along the same lines can be divided into two categories. The first
category optimises object layout and connectivity taking into account only
\emph{one} level of the memory hierarchy. These algorithms however are suitable
for use at runtime on arbitrary graphs. The second category are cache-oblivious
algorithms that can optimise data structure layouts without knowing the
precise hierarchy in use on the machine. Unfortunately cache-oblivious
algorithms have been designed only for specific data structures and the
techniques cannot be applied to graphs in general.

This paper proposes a Hierarchical Blocking Algorithm (HBA) as a solution. The
HBA proposed in this paper takes arbitrary graphs as input and produces a layout
that is sensitive to all levels of the memory hierarchy, information about
which is supplied to the algorithm. We show that not only does this make a large
difference to the processing of graphs, it also performs comparably to a cache
oblivious layout. The HBA therefore, closes the gap between cache oblivious and
cache sensitive (but limited to a single level) algorithms, an important
contribution of this paper.

The rest of this paper is organised as follows. We begin with some intuition
about HBA and describe how it is motivated by cache oblivious algorithms in
Section~\ref{sec:motivation}. We then describe a basic version of HBA applicable
only to trees in Sections~\ref{sec:theory}, \ref{sec:analysis} and \ref{sec:pract}. We then provide
extensions that make it applicable to arbitrary graphs in
Section~\ref{sec:simple} and extensions for space efficiency in
Section~\ref{sec:super}. We then describe implementation in three different
environments (custom graph processing in C, Boost C++ graph libraries and the
semispace garbage collector in the Jikes Java Virtual Machine) in
Section~\ref{sec:implementations}. We then evaluate HBA and show that it
delivers significant speedups in all these environments (from 10\% to as much as
21X). We then discuss related work and possible future extensions to HBA, before
concluding.

\section{Motivation and Intuition}
\label{sec:motivation}
The hierarchical blocking algorithm we have developed draws strong inspiration
from the van Emde Boas (VEB) layout. The VEB tree, originally proposed in a
paper outlining a design for a priority queue~\cite{veb} is an arrangement that
makes a tree data-structure cache oblivious i.e.\ likely to provide good
performance regardless of the cache hierarchy or units of spatial locality in operation.
The van Emde Boas(VEB) layout had
provided some of the initial motivation for this work.  

Figure~\ref{fig:veb} details the intuition behind the VEB layout. 
The VEB layout is a layout of a tree that is done by repeatedly splitting it
at the middle and \emph{recursively} laying out all the component subtrees in
contiguous units of memory. In the figure, the tree of depth $D$ is split into a
subtree (rooted at the original tree) of depth $\frac{D}{2}$ and this is
recursively laid out first. Next, the remaining subtrees, $O(2^\frac{D}{2})$ in
number, are laid out recursively. The VEB layout is complex to setup and maintain
for trees and difficult to apply to graphs in general. The first step in
applying it to a graph is to traverse the graph and prepare a sub-graph in the
form of a tree that covers it. This spanning tree could then be laid out in a
VEB layout. The key difficulty however is determining where to cut the tree
since a-priori knowing the diameter of a graph and its splits at runtime is a
difficult business. Further, the VEB layout does not consider heterogeneous
graphs where the objects representing graph vertices may have different sizes,
rendering it impractical to apply.

Our approach instead is to make the problem somewhat easier by assuming that the
memory hierarchy (of caches) is known at runtime as an input to the
algorithm. This can be used during the traversal to determine the spanning
subtree in conjunction with information about the size of in-memory
representations of graph vertices to determine the right split-point.

\begin{figure}
\centering
\scalebox{0.25}{\epsfig{file=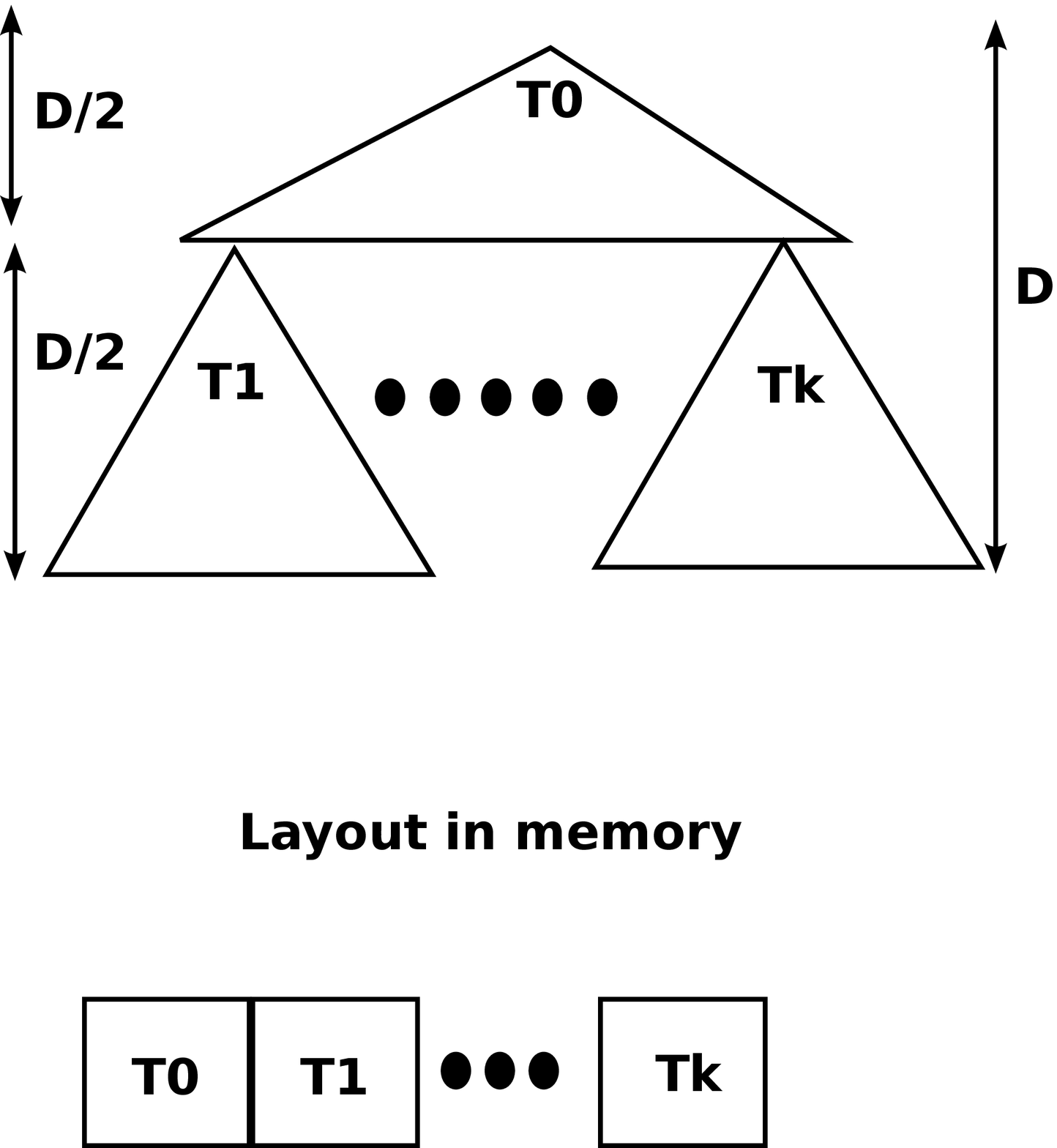}} 
\caption{van Emde Boas Layout}
\label{fig:veb}
\end{figure} 

\begin{figure}
\centering
\scalebox{0.25}{\epsfig{file=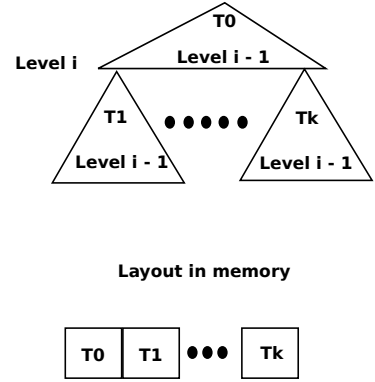}} 
\caption{Hierarchical Blocking Layout}
\label{fig:layout}
\end{figure}

Figure~\ref{fig:layout} shows graphically how this might be done. Assume an
algorithm $P_i$ that aims to copy a tree while traversing it, into blocks that
fit into the cache at level $i$. Using breadth first search, it can discover the
entire subtree that fits into a block at level $i$. It can then call breadth
first searches for individual subtrees that are rooted at the leaves of this
subtree (not shown in the Figure). For the subtree it has identified, it can
\emph{recursively} call $P_{i -1}$: an algorithm that can lay out a given tree
into blocks that fit into the cache at level $i-1$. This is shown in the Figure
and corresponds (roughly) to the recursive layout achieved by VEB. The key
difference is that we \emph{know} where to cut the spanning tree based on
runtime information about the memory hierarchy rather than simply using half the
diameter of the graph.

Having provided intuition behind our Hierarchical Blocking Algorithm (HBA), we
now proceed to discuss it in more detail below.

\section{Hierarchical Blocking for Trees}
\label{sec:theory}
In this section we develop a hierarchical blocking algorithm applicable to
trees. A tree is a graph (either directed or undirected) where every vertex has a
\emph{unique parent} and is itself connected to a number of \emph{children}.
We express the algorithm in terms of repeated breadth-first
searches~\cite{cormen} each of which is bounded to produce roots for new
searches.

We begin by introducing some basic notation that we use in this section. We
denote the application of an Algorithm $A$ to a graph vertex $x$ as $A(x)$, which
produces as output a list of vertices. We denote application of an algorithm $A$ to
a list of vertices $L$ by $A(L)$, which is done by applying it to \emph{each
individual vertex} and \emph{concatenating the outputs in order}. We denote the
repeated application of an algorithm (using the output of one as the input of
the next) as $A^k$, which means that we apply algorithm $A$ $k$ times. 

An important concept for hierarchical blocking is the space occupied by
representations of a vertex. For a vertex $v$ we assume a way to measure the
space occupied by the vertex, which we represent as $|v|$. This naturally extends
to applying an algorithm $A$ on input $I$ as $|A(I)|$, which is just
the sum of the spaces occupied by every vertex that is processed to produce the
output.

The core algorithm is bounded breadth first search, which we abbreviate
as $\text{BFS}_d$ that takes as input one vertex and produces a list of vertices that
are at distance $d$ from the start vertex. We measure distance as the number of
edges traversed. \emph{We also delegate to $\text{BFS}_d$ the job of copying traversed
vertices into a spatially contiguous unit of memory.}

We consider here a memory hierarchy of $n$ levels with monotonically
increasing units of spatial locality: $s_i$ for level $i$, with $s_i <
s_{i+1}$.  We now define a blocking algorithm that takes as input a list of
vertices and memory hierarchy levels, and recursively calls itself with
decreasing levels. We denote the algorithm for level $i$ as $P_i$ taking as
input a single vertex $o$. As indicated above, using $P_i$ on a list of vertices
naturally follows from the definitions given below.

\begin{itemize}
\item $P_1(o)$ = $\text{BFS}_d(o)$ where we choose depth $d$ such that :
\begin{itemize}
\item If $|\text{BFS}_0(o)| > s_1$ then $d = 0$
\item Else choose $d$ such that $|\text{BFS}_{d-1}(o)| < s_1$ and $\text{BFS}_{d}(o) \geq s_1$
\end{itemize}
\item $P_i(o) := P_{i-1}^k(o)$ where we choose $k$ such that:
\begin{itemize}
\item If $|P_{i-1}(o)| > s_i$ then set $k=1$
\item Else choose $k$ such  $|P_{i-1}^{k-1}(o)|< s_i$ and  $|P_{i-1}^{k}(o)|\geq s_i$
\end{itemize}
\end{itemize}
At the start we are given the root of the tree: $r$.  For hierarchical blocking
of the tree, we repeatedly apply $P_n$ starting from $r$ until we have copied all the vertices (the
output list is empty). The formalism given above produces exactly the layouts
that we have provided an intuition for in the previous section. A crucial point
to note here is that we allow the copying to overshoot the set limit by an
amount bounded by one application of the algorithm at the next underlying level.

\section{Analysis}
\label{sec:analysis}
A traversal of the tree needs to transfer blocks of size $s_i$ from the
$i^{\text{th}}$ level in the memory hierarchy. We now provide an upper bound on
the number of such blocks transferred. We make the observation that any
application of $P_n$ can be ultimately expressed as repeated applications of
$P_i$ for any $i < n$.  Consider a traversal of the copied tree generated by
$P_i(x)$ for an arbitrary vertex $x$ in the input graph. We consider a traversal
that starts from the copy of $x$ produced by $P_i(x)$ and terminates at some
leaf in the copied subtree produced by  $P_i(x)$.

For any memory hierarchy level $j$ this traversal leads to the transfer of some
number of blocks of size $s_j$.  Let an upper bound on the number of memory
blocks at level $j$ accessed due to this traversal be $B_i^j$, \emph{regardless
  of the start vertex}.

\begin{theorem}
$B_{i+1}^{i+1} = 2 + B_{i}^{i+1}$
\end{theorem}

\begin{proof}
For any $x$: $P_{i+1}(x)$ is defined as $P_{i}^d(x)$ with $|P_{i}^{d-1}(x)| \leq s_{i+1}$.
Now traversing a block of memory of size $s_{i+1}$ can incur accesses to at worst 2
blocks at level $i + 1$ (if the block start is not aligned).

The remaining part of the traversal is to a subtree produced by $P_{i}(y)$ for
some leaf $y$ of the previous traversal incurring at most $B_i^{i+1}$ block
transfers. Hence we have: $B_{i+1}^{i+1} = 2 + B_{i}^{i+1}$
\end{proof}

Under the common conditions where each level in the memory hierarchy is
sufficiently smaller than the next level and $B_1^1$ fits within $s_2$ we can
deduce a simple constant upper bound on $B_i^i$ for any $i$.

\begin{theorem}
If $\;\;\;\forall j\;\;\;B_1^1 < s_2 \;\;\;\text{and}\;\;\;4s_j \leq s_{j+1}$ then $ B_i^i \leq 4$
\end{theorem}

\begin{proof}
The theorem is true at $i = 1$ due to the conditions of the theorem where $B_1^1
< s_2$ and $s_2$ is at most 4 blocks of $s_1$. We now give an inductive proof.
Let the theorem be true till $i = k$. We have $B_{k+1}^{k+1} = 2 +
B_k^{k+1}$. From the induction hypothesis $B_k^k \leq 4$. Four blocks at level
$k$ is at most one misaligned block at level $k+1$ (due to the bounds on sizes
at each level) which is at most two aligned blocks at level $k+1$. Hence
$B_{k+1}^{k+1} \leq (2+2) = 4$
\end{proof}

In the context of the whole tree a traversal from root to leaf in the copied
tree incurs repeated costs of $B_n^i$ at memory hierarchy level $i$. If we
assume that $P_i(x)$ covers subtrees of depths at least $d_i$ then the number of
block accesses at memory hierarchy level $i$ for a traversal of depth $D$ is
bounded by $4\frac{D + 1}{d_i + 1}$. For a pseudorandom allocation of vertices
to memory, one would normally expect every access to cause a transfer, leading
to $D+1$ transfers. The hierarchical blocking algorithm is therefore able
to guarantee reduced transfers when $d_i > 3$.

Note that this is a pessimistic upper bound. For example, at the lowest level
(usually cache lines) any organisation that packs subtrees of depth one into a
cacheline leads to better performance with one cacheline serving two access
requests instead of one.

\section{Iterative Version}
\label{sec:pract}
We now present an iterative version of the hierarchical blocking algorithm
(HBA). In addition to being easier to understand, implement and analyse; it
forms the basis for extension to handle arbitrary graphs. The {\tt
HBTreeIterative} algorithm listed in Procedure~\ref{alg:reorg} is a direct
translation of the recursive algorithm described in the previous section. It
takes as input a root vertex and a description of a memory hierarchy and
performs runtime hierarchical blocking of the tree rooted at the supplied
vertex. Starting from this section, we introduce the term 'node', that we use as
an abstraction for the memory occupied by a graph vertex (and any associated
edge data structure, such as an edge list).

\begin{algorithm}
\caption{HBTreeIterative}
\label{alg:reorg}
\begin{algorithmic}[1]
\REQUIRE root: root to block from
\REQUIRE n: levels of hierarchy
\REQUIRE s[1..n]: block sizes (monotonically increasing)
\REQUIRE s[n+1] = INFINITY
\STATE Initialise to empty: Dequeue roots[1..n+1]
\STATE Initialise to empty: Dequeue leaves[1..n+1]
\STATE Initialise to zero: space[1..n+1] 
\STATE roots[n+1].push\_back(root)
\STATE level := n + 1
\LOOP
\IF{roots[level].is\_empty()}
\STATE ////////Refill
\STATE roots[level] := leaves[level]
\STATE leaves[level] := empty
\IF{space[level] $\geq$ s[level]}
\STATE //////Promote
\STATE leaves[level + 1].append(roots[level])
\STATE roots[level] := empty
\STATE space[level + 1] := space[level + 1] + space[level]
\STATE level := level + 1
\STATE continue loop 
\ENDIF
\ENDIF
\IF{roots[level].is\_empty()}
\STATE //////Promote
\IF{level = (n + 1)}
\STATE TERMINATE
\ELSE
\STATE space[level + 1] := space[level + 1] + space[level]
\STATE level := level + 1
\ENDIF
\ELSE
\STATE node := roots[level].pop\_front()
\IF{level $>$ 1}
\STATE //////Push work down
\STATE roots[level - 1].push\_back(node)
\STATE space[level - 1] := 0
\STATE level := level - 1
\ELSE
\STATE /////Do some copying work
\IF{UnconditionalCopyNode(node)}
\STATE space[1] := space[1] + sizeof(node)
\STATE leaves[1].append(children(node))
\ENDIF
\ENDIF
\ENDIF
\ENDLOOP
\end{algorithmic}
\end{algorithm}

\begin{algorithm}
\caption{UnconditionalCopyNode}
\label{alg:ucopy}
\begin{algorithmic}[1]
\REQUIRE node: node to copy
\STATE Copy node to tospace
\STATE tospace := tospace + sizeof(node)
\STATE Update parent of node in tospace to point to copy
\RETURN true
\end{algorithmic}
\end{algorithm}

The core data structures used in the algorithm are lists of {\tt roots} and
{\tt leaves} for each level of the hierarchy. In addition the {\tt space} array maintains the
amount of space used at each level. Lines 34--36 of the algorithm implement
$\text{BFS}_d$. This is done by taking the root node for the BFS, copying it
(through the call to {\tt UnconditionalCopyNode} and updating the space
used. The children produced during this BFS step are added to {\tt leaves[1]}.
If the amount of space used is less than the unit of spatial locality
for hierarchy level 1, all the leaves of the BFS are moved to {\tt roots[1]} and
a BFS step is subsequently performed for each of them to uncover their
children. Thus, only when the total space consumed at this lowest level is equal
to, or exceeds the unit of spatial locality at the lowest level, is the {\tt
  level} variable bumped and all the produced leaves of the BFS moved to level 2. It is easy to
see that this replicates the operation of $\text{BFS}_d$ described in the
previous section and discovers $d$ \emph{dynamically}.

For any $\text{level} = i >1$, the remaining steps of {\tt HBTreeIterative}
implement $P_i$. Recall that the input to $P_i$ is a list of nodes, this is held
in {\tt roots[i]}.  Lines 11-17 check whether repeated applications of $P_{i-1}$
have exhausted the unit of spatial locality given by $s[i]$. If so, the output
leaves are passed on to level $i+1$. Else, we repeatedly call $P_{i-1}$ on the
head of the $roots[i]$ list. Finally, the level {\tt n + 1} is simply a
placeholder for the output of $P_n$. For convenience, it is assumed to have
an infinite amount of spatial locality, i.e.\ it covers the whole memory.

Copying services are provided by {\tt UnconditionalCopyNode} that copies nodes
into a region of memory whose top is held in the {\tt tospace} variable. We
assume that this region of memory is infinite. A practical implementation of
{\tt UnconditionalCopyNode} could simply call a standard heap allocation
function (such as {\tt malloc}) to allocate memory to copy into, although this
inserts metadata before the copied object that can reduce locality (as we
discuss later).

\begin{figure}
\centering
\scalebox{0.6}{\epsfig{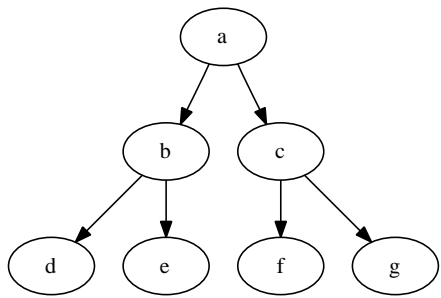}} 
\caption{Small Example Tree}
\label{fig:example}
\end{figure}

In order to better understand the operation of \\{\tt HBTreeIterative}, consider
the graph of Figure~\ref{fig:example}. If the input to {\tt HBTreeIterative} is
the node {\tt a} then the first thing the algorithm does is to add {\tt a} to
roots[n+1]. The node {\tt a} then bubbles down to roots[1] through repeated
iterations of the loop. It is then passed to \\{\tt UnconditionalCopyNode} at line
34. Next, line 36 adds its children, nodes {\tt b} and {\tt c} to {\tt
  leaves[1]}. Since {\tt roots[1]} is now empty, both {\tt b} and {\tt c} are
copied into {\tt roots[1]} at line 9. Assume for this example that at least one
node fits into {\tt s[1]} and so {\tt b} and {\tt c} are processed in turn 
resulting in {\tt roots[1]} containing {\tt d}, {\tt e}, {\tt f}, and {\tt
  g}. If now the algorithm finds {\tt space[1] > s[1]} then it promotes these
four nodes to level 2 and then calls $P_2$ on them (unless {\tt s[2]} is also
exhausted). This finally results in a consecutive layout in memory of nodes {\tt
  a}, {\tt b} and {\tt c}, followed by (partial) subtrees rooted at {\tt d},
{\tt e}, {\tt f} and {\tt g} produced by $P_2$.

\subsection{Complexity}
Given a tree with $N$ nodes and $E$ edges, the {\tt HBTreeIterative} algorithm
performs a graph traversal on it. It visits every edge exactly once (in line
36). For any node, after discovery, the node is added to every one of the root
and leaf lists \emph{at most once}. Hence, given $L$ levels in the memory
hierarchy the algorithm has a worst-case complexity of $O(LN + E)$. In a tree
the number of edges is one less than the number of nodes and hence the
complexity is $O(LN)$. Note that in
practice with settings such as $s[1]=64$ for and $s[2]=1024$ (used in this
paper) most nodes are only added to lists at levels 1 and 2 before being
processed and never make it to higher levels. Practically, this keeps the
overhead of the algorithm (and its variants) close to $O(N + E)$ or $O(N)$ for
trees.

\subsection{Limitations}

A key limitation of {\tt HBTreeIterative} is that it applies only to
trees. There are two reasons why it cannot be used on arbitrary graphs.  The
first is that {\tt UnconditionalCopyNode} expects that it is passed any node
exactly \emph{once}. This is easily violated in the case of multiple parents (as
in directed acyclic graphs) or graphs with cycles. A related problem arises
because more than one pointer may exist to a node and hence {\tt
  UnconditionalCopyNode} should be able to update parent pointers even if the
node has already been copied. In the next section we describe extensions to {\tt
  HBTreeIterative} that allow it to be used on arbitrary graphs.

\section{Extension to Arbitrary Graphs}
\label{sec:simple}
Extending {\tt HBTreeIterative} to arbitrary graphs first requires lowering our
level of abstraction somewhat. Along these lines, we introduce the notion of a
slot. A slot is simply a pointer to a node. Any given arbitrary graph is
therefore 'rooted' at multiple slots. Readers familiar with garbage collectors
in Java will notice that we have borrowed these two terms from
there~\cite{gc_book}.

Processing an arbitrary graph requires processing each root slot in turn. This
is done by calling {\tt InitiateCopy} in Procedure~\ref{alg:initcopy}. For every
given root slot, it calls {\tt HBTreeIterative} thereby processing individual
components of the graph. Note that we \emph{do not require graphs with different
  roots to be unreachable from each other}.

\begin{algorithm}
\caption{InitiateCopy}
\label{alg:initcopy}
\begin{algorithmic}[1]
\REQUIRE roots: queue of slots to start copying
\WHILE{not roots.is\_empty()}
\STATE slot := roots.pop\_front()
\STATE HBTreeIterative(*slot, levels, Sizes[1..levels])
\ENDWHILE
\end{algorithmic}
\end{algorithm}

The only change we make to {\tt HBTreeIterative} is to call {\tt
  ConditionalCopyNode} at line 34 instead of \\{\tt
  UnconditionalCopyNode}. Procedure~\ref{alg:ccopy} describes the former. The major change is the introduction of 
the Forward table that has an entry for \emph{each} possible node, indicating
whether that node has been forwarded. If not
already forwarded, it forwards (determines the position in tospace) the node and returns an appropriate
indication. {\tt HBTreeIterative} then uses the returned indication to ensure
that every node is considered \emph{at most once} thereby solving the problem of
multiple parents and cycles encountered in arbitrary graphs.

\begin{algorithm}
\caption{ConditionalCopyNode}
\label{alg:ccopy}
\begin{algorithmic}[1]
\REQUIRE node: node to copy
\IF{Forward[node] = UNFORWARDED}
\STATE Forward[node] := tospace
\STATE tospace := tospace + sizeof(node)
\RETURN true
\ELSE
\RETURN FALSE
\ENDIF
\end{algorithmic}
\end{algorithm}

\begin{algorithm}
\caption{CompleteCopy}
\label{alg:compcopy}
\begin{algorithmic}[1]
\REQUIRE roots: list of pointers to start copying
\WHILE{not roots.is\_empty()}
\STATE slot := roots.pop\_front()
\IF{not Copied[*slot]}
\STATE Copied[*slot] = true
\STATE Copy *slot to Forward[*slot]
\ENDIF
\STATE *slot := Forward[*slot]
\FORALL{child\_slot pointers in *slot}
\STATE roots.push\_back(child\_slot)
\ENDFOR
\ENDWHILE
\end{algorithmic}
\end{algorithm}

Since {\tt ConditionalCopyNode} no longer updates pointers or copies nodes, we  introduce a
post-processing phase called {\tt CompleteCopy}. This is shown in Procedure~\ref{alg:compcopy}
and is called after {\tt InitiateCopy} has completed. It traverses the graph
starting at the roots again and maintains a {\tt Copied} map to avoid copying a
node more than once. It also \emph{updates} all the slots to point to the copies
in {\tt tospace}.

\subsection{Complexity}
We now consider the complexity of {\tt InitiateCopy},\\ {\tt HBTreeIterative} and
{\tt CompleteCopy} taken together. Slots now explicitly represent edges in the
graph. Every slot (edge) is still considered at most once (or twice if it is a
root slot), this includes lookups in the extra maps. Any given node is also
processed at most once at copying and enters (and leaves) every one of the $L$
lists at most once. Hence the \emph{asymptotic} complexity of the algorithm
remains at $O(LN + E)$.

\subsection{Limitations}
The extensions to deal with arbitrary graphs suffer from two key problems. The
first problem naturally is the need to have two passes through the graph.
The second problem is the space cost of maintaining the extra
maps. A related problem that we have not thus far considered is the cost of
maintaining the root and leaf dequeues. Even maintained as linked lists (as we
do) they require one {\tt next} pointer per node. Note that the space overheads are
bounded by $O(N)$ and do not depend on the number of edges. Nevertheless, it is
desirable to try to eliminate them.

In spite of these limitations, we use an actual implementation of the generalised
HBA described in this section in the evaluation. For applications where offline
Reorganisation of large graphs is acceptable it is simple to implement and
effective. 

\section{Single-Pass and Possibly Metadata-Less Blocking}
\label{sec:super}
We now introduce the final and most sophisticated HBA. Before introducing the
algorithm, we make some observations about the operation of {\tt
  HBTreeIterative} in the case of general graphs. For any node that is copied,
all its unvisited children are added as a group to {\tt leaves[1]}. This group of nodes
continues \emph{unbroken} through various lists until it enters a {\tt roots}
list. After that they are dequeued in order to be bubbled down and copied. Note
that once a node is picked off a {\tt roots} list at line 26 of {\tt
  HBTreeIterative} it is copied immediately.

The key idea we take away from this observation is that it is possible to
\emph{represent this group} of nodes by its parent. Once the group (parent) enters a {\tt
  roots} list, instead of popping the parent, we pop \emph{slots in the parent}
one by one and bubble them down in turn to be processed. Processing the slot
involves both conditionally copying the target and updating the slot to point to
the new version of the node. We now introduce {\tt HBGraphOnePass} that
incorporates these ideas. Unlike the version for trees, it takes as input the
root slot to start processing from (and not the root node pointed to by that
slot). It still depends on (a slightly modified for interface reasons) {\tt
  InitiateCopy} to iterate through roots but eliminates {\tt CompleteCopy}.

\begin{algorithm}
\caption{HBGraphOnePass}
\label{alg:general}
\begin{algorithmic}[1]
\REQUIRE root\_slot: root slot to block from
\REQUIRE n: levels of hierarchy
\REQUIRE s[1..n]: block sizes (monotonically increasing)
\REQUIRE s[n+1] = INFINITY
\STATE Initialise to empty: Dequeue roots[1..n+1]
\STATE Initialise to empty: Dequeue leaves[1..n+1]
\STATE Initialise to zero: space[1..n+1] 
\STATE level := 1
\STATE /////Conditionally copy root slot
\STATE old\_node := *root\_slot
\IF{CopySlot(root\_slot)}
\STATE space[1] := space[1] + sizeof(*root\_slot)
\STATE leaves[1].append(old\_node)
\ENDIF
\LOOP
\IF{roots[level].is\_empty()}
\STATE ////////Refill
\STATE roots[level] := leaves[level]
\STATE leaves[level] := empty
\IF{space[level] $\geq$ s[level]}
\STATE //////Promote
\STATE leaves[level + 1].append(roots[level])
\STATE roots[level] := empty
\STATE space[level + 1] := space[level + 1] + space[level]
\STATE level := level + 1
\STATE continue loop 
\ENDIF
\ENDIF
\IF{roots[level].is\_empty()}
\STATE //////Promote
\IF{level = (n + 1)}
\STATE TERMINATE
\ELSE
\STATE space[level + 1] := space[level + 1] + space[level]
\STATE level := level + 1
\ENDIF
\ELSE
\STATE slot := roots[level].pop\_front\_slot()
\FORALL{i $>$ 1 and i $<$ level}
\STATE //////Init level
\STATE space[i] := 0
\ENDFOR
\STATE level := 1
\STATE /////Do some copying work
\STATE old\_node := *slot
\IF{CopySlot(slot)}
\STATE space[1] := space[1] + sizeof(*slot)
\STATE leaves[1].append(old\_node)
\ENDIF
\ENDIF
\ENDLOOP
\end{algorithmic}
\end{algorithm}

\begin{algorithm}
\caption{CopySlot}
\label{alg:slotcopy}
\begin{algorithmic}[1]
\REQUIRE slot: slot to copy
\STATE copied := false
\IF{Forward[*slot] = UNFORWARDED}
\STATE Forward[*slot] := tospace
\STATE tospace := tospace + sizeof(*slot)
\STATE copy *slot to Forward[*slot]
\STATE copied := true
\ENDIF
\STATE *slot := Forward[*slot]
\RETURN copied
\end{algorithmic}
\end{algorithm}

{\tt HBGraphOnePass} also uses a slightly different helper routine {\tt
  CopySlot} to complete copying of nodes. It directly updates the slot with the
copy of the node. Assuming that copying was required the node is then explored
for children. Note that in line 39 of {\tt HBGraphOnePass} we append the
\emph{old} node to the {\tt leaves[1]} list. This node is bubbled up and
ultimately moves to a {\tt roots} list. In line 30, instead of popping the node,
we pop its children one by one. This can be implemented by maintaining constant
sized state about which child has been popped. Note that we hoist a copy of part of the
processing for the root\_slot to lines 6--9.

Other than these changes {\tt HBGraphOnePass} operates similarly to {\tt
  HBTreeIterative}. To illustrate this, consider again the example graph in
Figure~\ref{fig:example}. The algorithm is passed the slot pointing to node {\tt
  a}.  It then copies node {\tt a} and adds {\tt a} \emph{itself} to {\tt
  leaves[1]}. Assuming {\tt s[1]} is not exhausted {\tt a} then moves to {\tt
  roots[1]}. Slots containing children of {\tt a} are then popped off by the
call to {\tt pop\_front\_slot} and this {\tt b} and {\tt c} are copied
next. Tracing the operation further, it should be evident that {\tt
  HBGraphOnePass} produces the same layout in memory as {\tt HBTreeIterative}
for the example.

\subsection{Complexity}
The complexity analysis for {\tt HBGraphOnePass} is substantially the same as in
the previous section. Every slot is considered at most once from lines 30--39 of
{\tt HBGraphOnePass} (other than being passed in as a root slot). Nodes traverse
every list at most once. Thus {\tt HBGraphOnePass} also has an asymptotic
complexity of $O(LN + E)$.

\subsection{Eliminating Metadata}
\label{sec:nomdata}
A simple observation also serves to eliminate the need for $O(N)$ extra
metadata. All nodes in a graph would have at least one pointer worth of space
(unless the graph is using a particularly compressed format). Further if the
graph is to represent any form of branching it would have space for at least two
pointers in its node representation. 

We use the first available pointer to store a pointer to the forwarded
copy. This eliminates to need for the {\tt   Forward} table since slots that
need copying point to old objects that can be looked up to determine the forward
pointer. In our implementations, we set the last bit to distinguish the forward
pointer from the same field in objects that have not yet been copied (since they
would point to objects aligned at 4 byte boundaries in our
implementation). Further, we can use the other
available pointer for manipulation of the lists representing the dequeues. This
eliminates the need for any external metadata, removing the need for $O(N)$
extra space.

Note that this elimination is made possibly by the organisation of {\tt
  HBGraphOnePass} that uses parent nodes to represent groups of children. In the
absence of this observation, we would have been forced to use dequeues of slots
in order to eliminate the extra pass thus rendering impossible elimination of
extra metadata for the dequeues. Further, this metadata-less one pass {\tt
  HBGraphOnePass} algorithm is a significant advancement over previous work.
Cheney~\cite{cheney} had shown that it was possible to use a breadth-first
traversal over objects in the heap without the need for any extra
metadata. Although Wilson~\cite{wilson} had developed hierarchical BFS for a
single level, it required one pointer per page of memory. {\tt HBGraphOnePass}
subsumes Wilson's algorithm as a special case of a single-level memory hierarchy
and also admits implementation without the need for extra metadata, similar to
Cheney's copying algorithm.

Of course, the technique described in this subsection is optional. For example
one could also allocate the extra $O(N)$ metadata directly in objects, such as
we have done for integration with the Jikes Java Virtual Machine~\cite{jikes}
garbage collector in one of our implementations. There, the forwarding pointer
already existed in the object header and we found it simplest to just add another
field for manipulation in the dequeue lists.

\section{Implementations}
\label{sec:implementations}
In this section we describe three different environments into which we have
integrated versions of our HBA. This section focuses on graph
representation and concrete interfaces for HBA. Another important focus area for
this section is memory allocation. Allocating target memory for copied nodes can
broadly follow one of two strategies. One is to use the system provided memory
manager that is already in use. This has the advantage of integrating cleanly
with existing code that uses graphs since there is no need to write an
additional memory manager and allocated nodes can be freed by the rest of the
application. A disadvantage to this approach is that system memory managers
(such as {\tt malloc}) introduce additional metadata at the head of each
object. This should be taken into account when calculating object sizes in the
HBA and additionally reduces the effectiveness of blocking in improving spatial
locality. Also, memory managers such as {\tt malloc} often use discontiguous
pools for objects of different sizes. This introduces further fragmentation if
graph nodes are differently sized, which is often the case for variable sized
edge lists attached to the nodes. The other option is to use a memory manager
with external metadata to manage to space copied to, which can introduce the
complexity of using multiple memory managers. 

We have implemented HBA in three environments for evaluation: custom graph
implementations written in C, as an add-on to the Boost Graph Library in C++ and
finally as a modification to the traversal phase of the semispace copying
collector in the Jikes Java Virtual Machine. We now discuss these
implementations individually.

\subsection{Boost C++ Graph Library}
\label{sec:boost}
The Boost C++ Graph Library (BGL)~\cite{boost} is a library for in-memory
manipulation of graphs. It makes extensive use of C++ generics (making extensive
use of templates) to provide a customisable interface for storing graph
structured data. We wrote an extension to the library that takes as input a
graph stored in the \emph{adjacency list} representation and produces a new
graph after hierarchical blocking. The adjacency list representation stores a
list of vertices in an iterable container, that allows one to iterate over every
vertex in a graph and then apply a suitable function. For each vertex the list
of edges originating at (for directed graphs) or connected to (for undirected
graphs) is maintained in another iterable container attached to the
vertex. Although our implementation is generic, we experimented with C++ STL
vectors as the container for both the types of components.

Our HBA addition to Boost uses the simpler version of the algorithm described in
Section~\ref{sec:simple}. It makes use of two external $O(N)$ sized vectors and
assumes a canonical mapping from the set of vertices to non-negative integers:
$f:V\rightarrow \{1\;\;..\;\;|V|\}$.  This is already provided by Boost and we
use this integer to index into the $O(N)$ sized vectors. The first vector is
used to maintain a ``next'' pointer ($f(\text{next})$ rather than
$\text{next}$). The second vector (we call it the {\tt RemapVector}) assigns,
for each given node in the input graph a number indicating its position in the
copied graph i.e\ if {\tt RemapVector(j) = RemapVector(i) + 1} then node {\tt j}
should be copied right after node {\tt i} in the output graph. It is easy to see
how the remap vector is set up by calls to {\tt ConditionalCopyNode}. We use the
produced remap vector in {\tt CompleteCopy} to actually produce the output
graph.

We use the memory allocator provided by Boost thereby incurring the overheads
described above. We have assumed for object size calculation that each edge
occupies an area equal to three pointers (for source, destination and
edge-weight information) and two pointers worth of memory management
metadata. Since the vertices are already laid out in a vector, we multiply the
number of outgoing edges by the space occupied by five contiguous pointers to
determine the object size for the HBA algorithm.

Our decision to use the simpler algorithm for integration into Boost was
guided in part by the observation that others have deemed the overheads of
storing all the vertices of a graph in memory acceptable~\cite{semiem}.

\subsection{Custom Graph Implementations in C} 
\label{sec:custom-c}
The Boost graph library introduces a number of overheads internal to the objects
used to represent graph vertices and edges, in part due to the need to be
generic and object-oriented. In order to explore the benefits that HBA can bring
to graphs constructed out of carefully designed minimal objects we also wrote
a custom implementation of binary search trees and undirected graphs for
evaluation. For the C implementations, we allocated a large chunk of memory to
copy the nodes into, this is done by maintaining the size of the graph as it is
loaded and calculating the total space required for the copy, in advance. This
eliminates all overhead due to memory manager metadata. It is easy to write a
memory manager that makes use of external metadata~\cite{vee} to manage this
space, although we have not done so for this implementation.

\subsubsection{Binary Search Trees}
We use the fairly minimalist representation of binary search tree nodes shown
below:
{\scriptsize
\begin{verbatim}
/* Basic bst node */
typedef struct node_st {
  unsigned long k; struct node_st *l, *r;
#ifndef NO_REORG_PTR
  struct node_st **reorg_next;
#endif  
}node_t;
\end{verbatim}}

We have explored both {\tt HBTreeIterative} (that makes use of the {\tt
  reorg\_next} pointer) as well as {\tt HBGraphOnePass} that eliminates that
pointer.

\subsubsection{Undirected Graphs}
We have also written representations for undirected graphs in C. These make use
of the data structures described below:

{\scriptsize
\begin{verbatim}
typedef struct node_st {
  int id;  struct node_st* neighbours[0];
} node_t;

node_t **node_vector;
int *neighbour_cnts;
int node_cnt;
\end{verbatim}}

The graph node structure contains an integer identifier and an array of pointers
to its neighbours. Since this is an undirected graph, two neighbouring nodes
point to each other through their {\tt neighbours} arrays. In addition, we have
a list of vertices ({\tt node\_vector}), a list of neighbour counts ({\tt
  neighbour\_cnts}) and finally the count of the total number of nodes in the
graph ({\tt node\_cnt}). Maintaining the size of the neighbours array outside the
data structure improves cache-line utilisation. Note that as with the Boost
implementation, we have an enumeration of vertices as integers that allows
indexing appropriate arrays. A point that might not be evident is that HBA also
results in better utilisation of linear arrays such as {\tt
  neighbour\_cnts}. This is because adjacent nodes are more likely to be placed
 close to each other in those arrays. This was key to our decision to move the
neighbour count out of the containing  {\tt node\_t} object.

Finally, we also wrote an extension to our implementation of undirected graphs
to illustrate a \emph{beneficial and powerful application of the HBA
algorithm}. We allow the writing out of the entire graph after HBA to disk. The
nodes are written out in the order that they are produced by HBA. Therefore,
reloading the nodes results in an in-memory representation of the graph that is
\emph{already blocked}. Although we show in our evaluation that the overheads of
HBA are tolerable enough to apply at runtime, this feature serves to illustrate
that offline HBA of in-memory graphs and storing the results in a persistent
manner is very much possible and eliminates the overheads of HBA when processing
static graphs.
 
\subsection{Jikes RVM}
We have also implemented HBA as an extension to the traversal phase of a
semispace copying collector in the Jikes Java Virtual Machine. We have done this
to illustrate the ability of HBA to operate in dynamic environments with varying
graphs. The fundamental idea of the semispace copying collector is to divide the
heap into two 'spaces'. At any instance of time only one space is 'active', and
is used to allocate objects. When a garbage collection is triggered, all
mutator threads (that can change the connectivity or contents of objects on the
heap) are stopped. A collector thread then runs a traversal phase that finds all
reachable objects on the heap using a depth-first search in the baseline
implementation and copies them (on discovery) to the other 'inactive'
space. After completion of this traversal, the 'inactive' and 'active' spaces
switch roles until the next garbage collection cycle.

We have modified the traversal phase to use HBA in order to copy objects with
appropriate clustering. We use the single pass {\tt HBGraphOnePass}
algorithm. Since the object-header used in the Jikes RVM had already allocated a
pointer to hold miscellaneous information, we used this pointer to hold
forwarding information. We added another pointer size field to hold next pointer
information for maintaining the dequeues. We use this slightly bloated object
representation as the baseline (without HBA) as we felt that this fairly
reflected the fact that we could have eliminated this overhead with extra
work. Our implementation in the Jikes RVM is at a prototype level only, in part
to avoid the complexity of refactoring the garbage collection classes to
implement an optimised version. For example, it is difficult to interrupt the
scanning of objects on the heap to determine slots. Hence we have a suboptimal
implementation of {\tt pop\_front\_slot} that simply uses an array of 32 slots
to hold the results of a complete scan. Any overflows from this array are
treated as new roots by the HBA implementation. Nevertheless, we found the
implementation adequate to demonstrate the feasibility of integrating HBA into
the garbage collector of a managed environment, thereby showing that is can be
used in such environments and on changing graphs.

\section{Performance Evaluation}
We evaluate HBA on an system equipped with an Intel i5-2400S CPU and 16 GB of
RAM. For uniformity (the JikesRVM is 32-bit only) we use 32 bit code and thus
are limited to using under 4GB of main memory. We know a-priori that the system
has the following caches in its memory hierarchy (with corresponding settings
for the HBA):

\begin{enumerate}
\item Various caches (L1 data, L2 and L3) with a 64 byte line size. We set {\tt s[1]
  = 64}.
\item Open-page mode DRAM that provides lower latencies for consecutive accesses
  to the same 1024 byte page. We set {\tt s[2] = 1024}.
\item TLB caching page table translations for 4KB pages. A TLB miss incurs
  significant penalty for page table walks. We set {\tt s[3] = 4096}.
\item TLB caching of super page translations. The OS (linux) clusters groups of
  pages into 2MB superpages to reduce consumption of page table and TLB
  entries. We set {\tt s[4] = 2097152}.
\end{enumerate}

It is unclear to us (as it would be to users of HBA in the field) about which
level is most likely to impact performance for a particular graph and particular
traversal algorithm. Hence we usually use full HBA with the settings above,
indicated as {\tt HBA(all)} in the evaluation. Occasionally we use HBA for only
a subset of levels: such as the page level {\tt HBA(4k)}, this is essentially
produces the layouts of Wilson's hierarchical BFS~\cite{wilson}.

\subsection{C, Binary Search Tree}
We first consider the performance of a binary search tree written in C. The tree
is setup to hold a contiguous integer keyspace and is then queried with random
keys. We are interested in the average query time (measured over a minute of
continuous queries) as the traversal is affected by the locality of the
nodes. We investigate the following layouts of tree nodes in memory:
1. Psedorandom layout 2. BFS layout ((such as would be produced by Cheney
  et. al.\ \cite{cheney}) 3. DFS layout, some researchers have suggested that this
  might be a better way to layout nodes for locality than BFS~\cite{stamos}
  4. VEB layout and finally 5. HBA layout using the algorithm in this paper.

The results are shown in Figure~\ref{fig:bst}. As expected the pseudorandom layout
performs the worst. BFS performs better than DFS. The best performance is
provided by HBA, \emph{which performs almost comparably to VEB}. At a tree depth
of 25 (~64 million nodes) using HBA reduces query time by approximately 54\%
while using BFS reduces query time by approximately 31\%. A notable feature of
the graph is the knee around the tree depth of 18. This is because beyond that
depth the tree no longer fits in the 6MB last level cache leading to a sudden
increase in query time.

\begin{figure}
\centering
\scalebox{0.25}{\epsfig{angle=-90, file=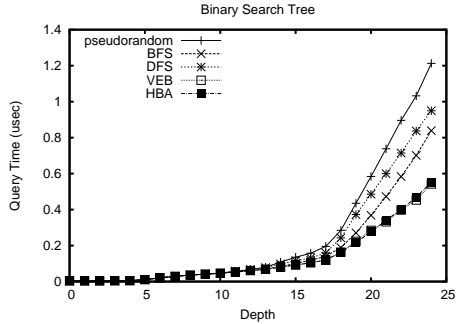}} 
\caption{Binary Search Tree Performance}
\label{fig:bst}
\end{figure}

Not every level of cache has an equal impact on performance. To illustrate this,
we ran the same experiment restricting HBA to various subsets of the memory
hierarchy. The results, shown in Figure~\ref{fig:bst_hba_vary} illustrate that
the cache (64 byte units of spatial locality) and the VM page (4K unit of
spatial locality) have the maximum impact on tree access.

\begin{figure}
\centering
\scalebox{0.25}{\epsfig{angle=-90, file=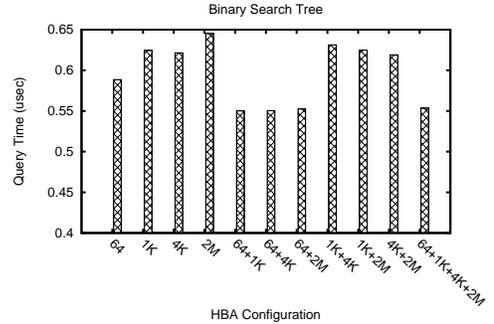}} 
\caption{Varying HBA parameters}
\label{fig:bst_hba_vary}
\end{figure}

Finally, we verify that HBA is indeed improving cache access. We used
cachegrind~\cite{cachegrind} to instrument the queries and simulate various
levels of the memory hierarchy of the actual system. Table~\ref{table:mr} shows
the miss rates for various levels of the hierarchy. Although DFS and BFS both
improve miss rates, HBA is most effective at reducing miss rates, explaining its
better performance. It is interesting to note that HBA when run over all levels
reduces miss rates for any level to that produced by running HBA to block only
for that level. This is an important result, since it illustrates that HBA
provides additive benefits for all the memory hierarchy levels it is aware
of. Finally, we note that HBA is almost as effective at tackling miss rates as
the cache-oblivious VEB layout.

\begin{table}
{\scriptsize
\begin{tabular}{|c||c|c|c|c|}\hline
     & L1d line & DRAM page & VM page & VM superpage \\\hline
Pseudorandom   &   0.42   &  0.53     &  0.44      & 0.42 \\\hline\hline 
BFS   &   0.36   &  0.43     &   0.26   & 0.07 \\\hline 
DFS & 0.33     & 0.23      & 0.19 & 0.11 \\\hline
VEB & \textbf{0.24}     & \textbf{0.15}      & \textbf{0.05} & \textbf{0.02} \\\hline\hline
64          & \textbf{0.25}     & 0.19  & 0.10 & 0.03\\\hline
1K           & 0.31     & \textbf{0.19}      & 0.07  & 0.02 \\\hline
4K           & 0.32     & 0.25      & \textbf{0.07} & 0.02\\\hline
2M           & 0.33     & 0.34    & 0.14 & \textbf{0.02} \\\hline
64+1K+4K+2M & \textbf{0.25}     & \textbf{0.17}  & \textbf{0.06} & \textbf{0.02}\\\hline
\end{tabular}}
\caption{Cachegrind miss rates}
\label{table:mr}
\end{table}

Finally, we illustrate the effect of removing HBA related metadata from the tree
node, using the {\tt HBGraphOnePass} algorithm and reusing pointer fields from
the old version of the object (as discussed in Section~\ref{sec:nomdata}). The
results shown in Figure~\ref{fig:bstnomdata} illustrate the improvements
obtained due to the lower memory footprint of the tree nodes:
approximately 14\% lower than that with an extra pointer per tree node.

\begin{figure}
\centering
\scalebox{0.25}{\epsfig{angle=-90, file=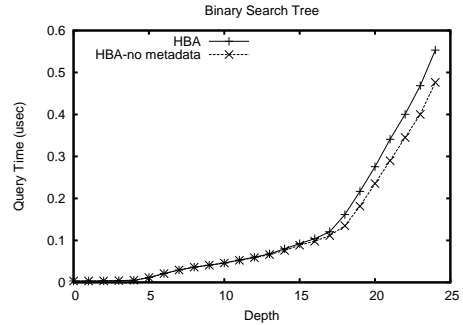}} 
\caption{Binary Search Tree: Effect of removing HBA metadata}
\label{fig:bstnomdata}
\end{figure}

\subsection{Arbitrary Graphs}
Trees represent an ideal workload for HBA, since they correspond exactly to the
spanning tree built during traversal. In this section, we consider more complex
graphs with a large number of connections. We use a synthetic graph generator
that is part of the SNAP suite~\cite{snap}. We consider various kinds of graphs
that are of current interest to the research community involved in mining
information from graph structured data:

\begin{enumerate}
\item Watts-Strogatz small world model~\cite{watts} (10 million nodes, 29
  million edges): These graphs have 
  logarithmically   growing diameter and model small world networks, such as
  social networks with the informally well known ``six degrees of separation''.  
\item Barabasi-Albert model~\cite{albert} (10 million nodes, 39 million edges) also models real world phenomena but
  provides graphs  where the out-degree of nodes follows a power law
  distribution. This is often  the case, for example, with web pages that link
  to each other.
\item 2d mesh (9 million nodes, 17 million edges) models real world road networks. Answering real time navigational
  queries on such networks are often a component of many online services.
\item 4ary tree (10 million nodes, one less edge). To provide some perspective on binary trees considered thus
  far, we also measure performance on trees where each node has 4 children.
\end{enumerate}

We use two different algorithms (in two different implementations). The first is
a single source shortest path algorithm (using Dijkstra~\cite{cormen}) that finds shortest paths from a given
source to all nodes in the graph. We use a random assignment of wights to edges
(a uniform random choice over  a range of size the same as the number of
vertices). We use an implementation of this algorithm in Boost
(Section~\ref{sec:boost}). We then turn the input graph into an undirected and
unweighted version by adding a reverse edge for every given edge and performing
a breadth-first search in our custom C environment (Section~\ref{sec:custom-c}).
The choice of algorithms and datasets is fairly similar to other approaches that
evaluate the performance of graph processing solutions~\cite{semiem}.

For each algorithm, data set and environment we measure the speedup of the
algorithm after HBA on the graph using both HBA(4k) as well as HBA(all) i.e.\
blocking for VM pages and for all levels of the hierarchy respectively. The
results shown in Table~\ref{table:graph} \emph{underscore the efficacy of HBA for
arbitrary graphs}. Large speedups (as high as 21X !) are obtained with
HBA. Speedups are generally higher for our custom C environment due to the
optimised (reduced) object footprints. Optimising for all levels of the memory
hierarchy often provides better performance than just optimising for one level,
underscoring the importance of a multi-level blocking algorithm. The results
also illustrate an interesting example of destructive interference between
levels. HBA(4k) performs slightly better than HBA(all) in the case of 2d
meshes. Other than this example, we have found that in all cases HBA(all)
performs at least as well as HBA(4k).

\begin{table}
{\scriptsize
\begin{tabular}{|c||c|c||c|c|}\hline
\multirow{2}{*}{Graph} & \multicolumn{2}{c||}{SSSP} & \multicolumn{2}{c|}{BFS} \\\cline{2-5}
 & HBA(4K) & HBA(all) & HBA(4k) & HBA(all)\\\hline
Watts\_Strogatz   &   1.41   &  1.44     &  1.38      & 1.40 \\\hline
Albert\_Barabasi   &   1.01   &  1.02     &   1.09   & 1.11 \\\hline 
2d\_mesh & 2.38     & 2.35      & 3.85 & 3.80 \\\hline
4ary\_tree & 1.00     & 1.00      & \textbf{20.70} & \textbf{21.31} \\\hline
\end{tabular}}
\caption{Graph Processing Speedups}
\label{table:graph}
\end{table}

\subsection{JikesRVM}
Our final set of results are from the Jikes Java Virtual machine. We measured
the performance of the same binary search tree considered in the C environment,
when implemented in Java. We configured the JVM to use a 1GB heap for the
experiments. We also configured it to perform a system-wide GC before starting
the query phase of the test. The results, shown in Figure~\ref{fig:bst_jikes}
indicated that the benefits seen with C are also replicated in the Java
environment. HBA for all levels with a tree depth of 22 provides a 29\% speedup
over the baseline version, while HBA for only the VM page provides a 19\%
speedup over the baseline version. Note that JVM memory limitations meant we
were unable to build trees of larger depth.

\begin{figure}
\centering
\scalebox{0.25}{\epsfig{angle=-90, file=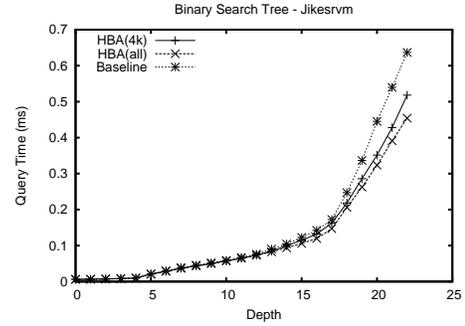}} 
\caption{Binary Search Tree: JikesRVM Performance}
\label{fig:bst_jikes}
\end{figure}

In a runtime environment the overhead of collection is also an important
factor. With this in mind, we measured the time for a semispace copy of the
entire heap after the tree has been completely built. The results are shown in
Figure~\ref{fig:bst_jikes_overhead}.  In the worst case, HBA adds an overhead of
only 18\%. Crucially the overhead of optimising for \emph{all} levels is the
worst case only 10\% more than optimising for only the page level. The average
overheads are much lower, well under 10\%.

\begin{figure}
\centering
\scalebox{0.25}{\epsfig{angle=-90, file=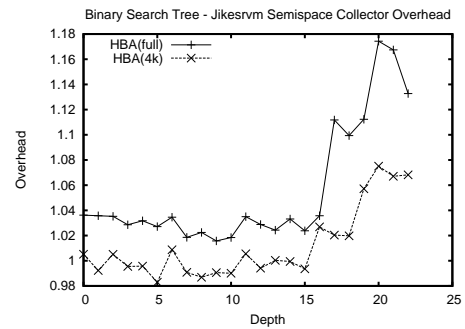}} 
\caption{Binary Search Tree: Semispace collector overheads}
\label{fig:bst_jikes_overhead}
\end{figure}

Finally we also measure overheads and performance with the more general DaCapo
benchmark suite~\cite{dacapo}. The results shown in
Fig~\ref{fig:dacapo_overhead} indicate that the overhead of adding HBA to the
garbage collector is under 15\% in all cases and usually under 10\%. In addition
for two of the benchmarks (antlr and fop) we see improved performance due to
more locality on the heap.

\begin{figure}
\centering
\scalebox{0.25}{\epsfig{angle=-90, file=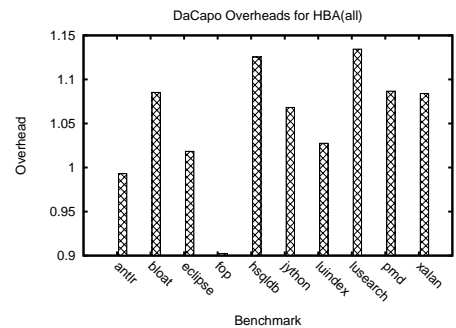}} 
\caption{DaCapo Benchmarks: Semispace collector overheads}
\label{fig:dacapo_overhead}
\end{figure}

\section{Related Work}
\label{sec:related}
There is a large body of existing research into improving the cache performance
of in-memory data. Broadly the approaches can be divided into three
classes. 

The first class of techniques deal with prefetching objects ahead of
use. An example of this is the approach of Luk et al.\ \cite{rds_prefetch}, who
place a prefetch pointer in linked list nodes to prefetch later nodes early
enough to avoid cache miss penalties during traversals. 
Dynamic approaches are also possible such as that of Chilimbi et
al.\ who profile a program to detect frequently occurring streams of accesses~\cite{chilimbi_hot}.

A second class of techniques is to statically modify the data structures
themselves to make them more cache friendly. One way is to use knowledge of the
cache hierarchy and transfer units to size data structure nodes, such as in
B-trees ~\cite{cormen}. This can be extended to make the B-tree nodes cache
friendly at various levels (similar to the objective of this work). Kim et al.\
\cite{fastbtree} extend the basic idea of B-trees to be architecture sensitive
at various levels using hierarchical blocking. Although their hierarchical
blocking produces layouts similar to our reorganisation algorithm they have a
static data structure redesign for B trees unlike our dynamic general purpouse
algorithm. Another approach to data structure design is cache oblivious data structures. These are
designed so as to improve spatial locality regardless of the level of memory
hierarchy and block size being considered. The ``van Emde Boas''
layout~\cite{veb}  forms the basis for many cache oblivious designs including
those for cache oblivious B-trees~\cite{streamingbtree}.  

A third class of techniques (including the one in this paper) are used at
runtime. One approach is to control memory allocation. Chilimbi et al.\
\cite{chilimbi_cc_layout} investigated the use of a specialised memory allocator
that could be given a hint about where to place the allocated node.
Another approach is to use the data structure traversal done by garbage
collectors to copy objects into new cache friendly
locations~\cite{copying}. Mark Adcock in his PhD thesis~\cite{adcock_phd}
considered a range of runtime data movement techniques including those triggered
by pointer updates. However none of these techniques consider the effect of
multiple units of locality in the memory hierarchy and in that sense this work
is orthogonal to all of them. It is possible to take the algorithm in this paper
and use it to improve on all of these locality maximisation techniques, which
are usually restricted to plain breadth-first search to discover nodes.

\section{Future Work}
The current implementation of HBA ignores the last level in a usual memory
memory hierarchy: persistent storage. It is extremely easy
to add a 512 byte sector size to HBA to also optimise layouts for transfer from
disk. Although we have not investigated this aspect yet, we believe HBA can also
significantly improve access to the last (persistent) level in typical memory
hierarchies.

\section{Conclusion}
We have presented a hierarchical blocking algorithm (HBA) that takes as input an
arbitrary graph and a description of a memory hierarchy and lays out graph nodes
to be sensitive to and provide better performance for that hierarchy. We have
investigated implementations of HBA in various settings and shown that it provides
non-trivial benefits in all of them; making the case the graph layout and memory
hierarchy sensitivity are important factors in the performance of graph
algorithms.

\bibliographystyle{plain}
\bibliography{paper}
\end{document}